\documentclass[letterpaper, 10 pt, conference,twocolumn]{ieeeconf}
\IEEEoverridecommandlockouts

\usepackage{amsmath, amsthm, amssymb, amsfonts, mathtools, xargs, tensor, units, cite, stmaryrd, mathrsfs, algpseudocode, comment, graphicx, url}
\usepackage[normalem]{ulem}
\usepackage[utf8]{inputenc}
\usepackage{multirow}
\usepackage{subcaption}
\usepackage[font=footnotesize]{caption}
\usepackage{epsfig}
\usepackage{epstopdf}
\usepackage[unicode=true, bookmarks=true, bookmarksnumbered=false, bookmarksopen=false, breaklinks=false, pdfborder={0 0 1}, backref=false, colorlinks=false, hidelinks]{hyperref}
\makeatletter
\let\NAT@parse\undefined
\makeatother
\newcommand{\dif}{\mathop{}\!\mathrm{d}}
\usepackage[colorinlistoftodos]{todonotes}
\theoremstyle{definition}
\newtheorem{definition}{Definition}

\theoremstyle{definition}

\newtheorem{theorem}{Theorem}

\theoremstyle{remark}
\newtheorem{remark}{Remark}

\DeclareMathOperator{\rank}{rank}

\begin{document}
	\title{Online Observer-Based Inverse Reinforcement Learning}
	\author{Ryan Self, Kevin Coleman, He Bai, and Rushikesh Kamalapurkar\thanks{The authors are with the School of Mechanical and Aerospace Engineering, Oklahoma State University, Stillwater, OK, USA. {\tt\footnotesize \{rself, kevin.coleman10, he.bai, rushikesh.kamalapurkar\}@okstate.edu}. This research was supported, in part, by the National Science Foundation (NSF) under award number 1925147. Any opinions, findings, conclusions, or recommendations detailed in this article are those of the author(s), and do not necessarily reflect the views of the sponsoring agency.}}
	\maketitle
\begin{abstract}
    In this paper, a novel approach to the output-feedback inverse reinforcement learning (IRL) problem is developed by casting the IRL problem, for linear systems with quadratic cost functions, as a state estimation problem. Two observer-based techniques for IRL are developed, including a novel observer method that re-uses previous state estimates via history stacks. Theoretical guarantees for convergence and robustness are established under appropriate excitation conditions. Simulations demonstrate the performance of the developed observers and filters under noisy and noise-free measurements.
\end{abstract}
\section{Introduction}
Inverse Reinforcement Learning (IRL) \cite{SCC.Russell1998,SCC.Ng.Russell2000,SCC.Abbeel.Ng2010}, sometimes referred to as Inverse Optimal Control \cite{SCC.Kalman1964}, is a subfield of Learning from Demonstration (LfD) \cite{SCC.Schaal1997} where the goal is to uncover a reward (or cost) function that explains the observed behavior (i.e., input and output trajectories) of an agent. Early results on IRL assumed that the trajectory of the agent under observation is truly optimal with respect to the unknown reward function \cite{SCC.Ng.Russell2000}. Since optimality is in general a strong assumption in a variety of situations, e.g., human operators and trajectories affected by noise or disturbances, IRL is extended to the case of suboptimal demonstrations (i.e., the case where observed behavior does not necessarily reflect the underlying reward function) \cite{SCC.Ziebart.Maas.ea2008}. While IRL has been an active area of research over the past few decades \cite{SCC.Ratliff.Bagnell.ea2006,SCC.Wulfmeier.Ondruska.ea2015,SCC.Li.Kiseleva.ea2019,SCC.Brown.Goo.ea2019,SCC.Levine.Popovic.ea2010,SCC.Levine.Popovic.ea2011,SCC.Sosic.KhudaBukhsh.ea2017,SCC.Wang.Klabjan2018a,SCC.Michini.How2012}, most IRL techniques are offline and require a large amount of data in order to uncover the true reward function.

Inspired by recent results in online Reinforcement Learning methods \cite{SCC.Vamvoudakis.Lewis2010,SCC.Wang.Liu.ea2016,SCC.Kamalapurkar.Walters.ea2018}, IRL has been extended to online implementations where the objective is to learn from a single demonstration or trajectory \cite{SCC.Molloy.Ford.ea2018,SCC.Kamalapurkar2018,SCC.Self.Harlan.ea2019a,SCC.Self.Mahmud.ea2020}. In \cite{SCC.Kamalapurkar2018,SCC.Self.Harlan.ea2019a}, batch IRL techniques are developed to estimate reward functions in the presence of unmeasureable system states and/or uncertain dynamics for both linear and nonlinear systems. The case where the trajectories being monitored are suboptimal due to an external disturbance is addressed in \cite{SCC.Self.Abudia.ea2020}, and \cite{SCC.Self.Mahmud.ea2020} estimates a feedback policy and generates artificial data using the estimated policy to compensate for the sparsity of data in online implementations. However, results such as \cite{SCC.Molloy.Ford.ea2018,SCC.Kamalapurkar2018,SCC.Self.Harlan.ea2019a,SCC.Self.Abudia.ea2020,SCC.Self.Mahmud.ea2020}, either require full state feedback, or rely on state estimators that require dynamical systems in Brunovsky Canonical form. In addition, none of the aforementioned online IRL methods address uncertainty in the state and control measurements.

This paper builds on the authors' previous work in \cite{SCC.Self.Abudia.ea2020,SCC.Self.Mahmud.ea2020}, where concurrent learning (CL) update laws are utilized to estimate reward functions online using output feedback. However, the dynamical systems in \cite{SCC.Self.Abudia.ea2020,SCC.Self.Mahmud.ea2020} are required to be in Brunovsky canonical form, and as such, only the output feedback case where the state is comprised of the output and its derivatives is addressed. In contrast, the IRL observer (IRL-O) technique in this paper generalizes to any observable linear system, since the developed IRL-Os are in a standard observer form where the state estimates are modified based on the innovation (i.e., the error between the actual and the estimated output). As a result, in the case of noisy measurements, they can be implemented as Kalman filters by using the Kalman gain, instead of the developed Lyapunov-based gain design, to select the observer gain. While stability of the filters in the case where the measurements are noisy is not studied in this paper, simulation results demonstrate that the IRL-Os utilizing both the Lyapunov-based gains and the Kalman filter gain are robust to measurement noise.

This paper details two IRL-O formulations. The first method, called the IRL memoryless observer (MLO), is similar to a standard Luenberger observer with a modified observer gain, and guarantees parameter convergence under a persistence of excitation (PE) condition. The second observer implements a novel idea of re-using previous system state estimates and control measurements, along with the Hamilton-Jacobi-Bellman equation, to gain insights into the quality of the current estimate of the reward function. The key advantage of the IRL history stack observer (HSO) over MLO is that it provides an additional guarantee for boundedness of the estimation errors under \emph{finite} (as opposed to \emph{persistent}) excitation~\cite{SCC.Rotithor.Trombetta.ea2019}. 

In summary, unlike the IRL results in \cite{SCC.Self.Abudia.ea2020} and \cite{SCC.Self.Mahmud.ea2020} that require systems in Brunovsky canonical form, the IRL methods developed in this paper are applicable to general observable linear systems. Contrary to traditional adaptive observers that require PE for stability and convergence, the novel memory-based HSO formulation developed in this paper guarantees boundedness of the weight estimates under loss of excitation. In addition, as evidenced by the simulation results, the technique developed in this paper, when implemented using the Kalman gain, is more robust to measurement noise than current methods such as \cite{SCC.Self.Abudia.ea2020} and \cite{SCC.Self.Mahmud.ea2020}.

\section{Problem formulation}\label{sec:NonLinearProblem}
	Consider an agent under observation with the following linear dynamics
	\begin{equation}
	\dot{x}=Ax+Bu, \ \ \ 	y^\prime = Cx,\label{eq:Nonlinear Problem Formulation}
	\end{equation}
    where $x:\mathbb{R}_{\geq 0}\rightarrow \mathbb{R}^n$ is the state, $u:\mathbb{R}_{\geq 0}\rightarrow \mathbb{R}^m$ is the control, $A\in\mathbb{R}^{n\times n}$ and $B\in\mathbb{R}^{n\times m}$ are constant system matrices, $y^\prime\in\mathbb{R}^{L}$ are the outputs, and $C\in\mathbb{R}^{L\times n}$ denotes the output matrix\footnote{For $a\in\mathbb{R},$ the notation $\mathbb{R}_{\geq a}$ denotes the interval $\left[a,\infty\right)$ and the notation $\mathbb{R}_{>a}$ denotes the interval $\left(a,\infty\right)$.}.
	
	The agent under observation is using the policy which minimizes the following performance index
	\begin{equation}
	J(x_0,u(\cdot)) = \int_{0}^{\infty}\left(x(t)^TQx(t)+u(t)^TRu(t)\right)\dif t,\label{eq:reward}
	\end{equation}
    where $x(\cdot;x_0,u(\cdot))$ is the trajectory of the agent generated by the optimal control signal $u(\cdot)$ starting from the initial condition $x_0$. The objective of this paper is to estimate the unknown matrices $Q$ and $R$ by utilizing input-output pairs.
	
	\begin{remark}
	Since $Q$ and $R$ can be selected to be symmetric without loss of generality, the developed IRL method only estimates the elements of $Q$ and $R$ that are on and above the main diagonal.
	\end{remark}

\section{Inverse Reinforcement Learning}\label{sec:IBE}
    Under the premise that the observed agent makes optimal decisions, the state and control trajectories, $x(\cdot)$ and $u(\cdot)$, satisfy the Hamilton-Jacobi-Bellman (HJB) equation \cite{SCC.Liberzon2012}
	\begin{equation}
	H\left(x\left(t\right),\nabla_{x}\left(V^{*}\left(x\left(t\right)\right)\right)^{T},u\left(t\right)\right)=0,\forall t\in \mathbb{R}_{\geq 0},\label{eq:inverse HJB}
	\end{equation}
	and the optimal control equation
	$
	    u(x(t))=-\frac{1}{2}R^{-1}B^T\nabla_x\left(V^{*}\left(x\left(t\right)\right)\right)^{T},\label{eq:optimal control}
	$

    where $ V^{*}:\mathbb{R}^{n}\to \mathbb{R} $ is the unknown optimal value function and $H:\mathbb{R}^n\times \mathbb{R}^n \times \mathbb{R}^m \to \mathbb{R}$ is the Hamiltonian, defined as $H(x,p,u):=p^T\left(Ax+Bu\right)+x^TQx+u^TRu$. Given a solution $S$ of the Algebraic Riccati Equation, the optimal value function can be calculated as $V^*(x) = x^T S x$.
	
	To aid in the estimation of the reward function, note that $V^*$, $x^TQx$, and $u^TRu$ can be linearly parameterized as $ V^*(x)=\left({W}_{V}^*\right)^{T}\sigma_V(x)$, $ x^TQx = \left({W}_{Q}^*\right)^{T}\sigma_Q(x)$, and $ u^TRu = \left({W}_{R}^*\right)^{T}\sigma_{R1}(u)$, respectively, where  $\sigma_V(x):\mathbb{R}^{n}\to\mathbb{R}^P$, $\sigma_Q(x):\mathbb{R}^{n}\to\mathbb{R}^P$, and $\sigma_{R1}(u):\mathbb{R}^{m}\to\mathbb{R}^M$, are the basis functions, selected as
	\begin{align*}
	    &\sigma_V(x)=\sigma_Q(x):=[x_1^2,2x_1x_2,2x_1x_3,\ldots,2x_1x_n,x_2^2,\nonumber\\& \ \ \ \ \ \ \ \ \ \ \ \ \ \ \ \  2x_2x_3,2x_2x_4,\ldots,x_{n-1}^2,\ldots,2x_{n-1}x_n,x_n^2]^T,\nonumber\\
	    &\sigma_{R1}(u):=[u_1^2,2u_1u_2,2u_1u_3,\ldots,2u_1u_m,u_2^2,\nonumber\\&\ \ \ \ \ \ \  \ \ \ \ \ \  2u_2u_3,2u_2u_4,\ldots,u_{m-1}^2,\ldots,2u_{m-1}u_m,u_m^2]^T,
	\end{align*}
	and 
	$W_V^*\in\mathbb{R}^P$, $W_Q^*\in\mathbb{R}^P$, and $W_R^*\in\mathbb{R}^M$, are the ideal weights, given by
	\begin{align*}
		W_V^*&=\left[S_{11},S_{1}^{(-1)},S_{22},S_2^{(-2)},\ldots,S_{n-1}^{-(n-1)},S_{nn}\right]^T,\\
		W_Q^*&=\left[Q_{11},Q_{1}^{(-1)},Q_{22},Q_2^{(-2)},\ldots,Q_{n-1}^{-(n-1)},Q_{nn}\right]^T,\\
		W_R^*&=\left[R_{11},R_{1}^{(-1)},R_{22},R_2^{(-2)},\ldots,R_{m-1}^{-(m-1)},R_{mm}\right]^T.
	\end{align*}
	where, for a given matrix $E\in\mathbb{R}^{n\times n}$, $E_{ij}$ denotes the corresponding element in the $i-$th row and the $j-$th column of the matrix $E$, and $E_i^{(-j)}$ denotes the $i-$th row of the matrix $E$ with the first $j$ elements removed, i.e., $E_3^{(-3)}:=\left[E_{34},E_{35},\ldots,E_{3(n-1)},E_{3n}\right].$
	
	Using $ \hat{W}_{V} $, $ \hat{W}_{Q} $, and $ \hat{W}_{R} $, which are the estimates of $ W_{V}^{*} $, $ W_{Q}^{*}$, and $ W_{R}^{*}$, respectively, in \eqref{eq:inverse HJB}, the inverse Bellman error (IBE) $ \delta^{\prime}:\mathbb{R}^{n}\times \mathbb{R}^{m}\times\mathbb{R}^{2P+M}\to\mathbb{R} $ is obtained as
	$
	\delta^{\prime}\left({x},u,\hat{W}^\prime\right)=\hat{W}_{V}^{T}\nabla_x\sigma_V(x)
	\left(Ax+Bu\right)
	+\hat{W}_{Q}^{T}\sigma_Q(x)+\hat{W}_R^T\sigma_{R1}(u),\label{eq: IBE}
	$
	where $\hat{W}^\prime:=\begin{bmatrix}\hat{W}_V^T & \hat{W}_Q^T & \hat{W}_R^T\end{bmatrix}^T.$
	
	The vector $Ru$ can be linearly parameterized as $Ru = \sigma_{R2}(u)W_R^*$, where $W_R^*$ is as previously defined in the IBE and $\sigma_{R2}(u):\mathbb{R}^m\to\mathbb{R}^{m\times M}$, where the features $\sigma_{R2}(u)$ can be explicitly calculated as shown in \eqref{eq:u_basis}.
	
	\begin{table*}[ht]
		\begin{equation}
			\sigma_{R 2}(u)=
			\begin{bmatrix}
				u^{T} & 0_{1 \times m-1} & 0_{1 \times m-2} & \ldots & 0 \\
				u_{(1)}e_{2,m} & \left(u^{(-1)}\right)^{T} & 0_{1 \times m-2} & \ldots & 0 \\
				u_{(1)}e_{3,m} & u_{(2)}e_{2,m-1} & \left(u^{(-2)}\right)^{T} & \ldots & 0 \\
				\vdots & \vdots& \vdots & \ddots & \vdots \\
				u_{(1)}e_{m,m} & u_{(2)}e_{m-1,m-1} & u_{(3)}e_{m-2,m-2} & \cdots & \left(u^{-(m-1)}\right)^{T}
			\end{bmatrix}.\label{eq:u_basis}
		\end{equation}
		\caption{Explicit characterization of $ \sigma_{R 2}(u) $. For a given vector $u\in\mathbb{R}^{1\times m}$, $u^{(-j)}$ denotes the vector $u$ with the first $j$ elements removed and $e_{i,j}$ denotes a row vector of size $j$, with a one in the $i-$th position and zeros everywhere else, and $ u_{(i)} $ denotes the $i-$th element of $u$}
	\end{table*}
	
	Using $\hat{W}_R$ and $\hat{W}_V$ for $W_R^*$ and $W_V^*$, respectively, in the optimal controller equation $2Ru=-B^T\nabla_x\left(V^{*}\left(x\right)\right)^{T}$, a control residual error $\Delta^\prime_u:\mathbb{R}^n\times\mathbb{R}^m\times\mathbb{R}^{2P+M}\to\mathbb{R}^m$ is obtained as
	$
	    \Delta^\prime_u(x,u,\hat{W}^\prime)=B^T\left(\nabla_x\sigma_V(x)\right)^T\hat{W}_V+2\sigma_{R2}(u)\hat{W}_R.
	$
	
	Augmenting the control residual error and the inverse Bellman error yields the error equation
	\begin{equation}\label{eq: Initial IRL Matrix}
	    \begin{bmatrix}\delta^\prime\left(x,u,\hat{W}^\prime\right)\\\Delta^\prime_u\left(x,u,\hat{W}^\prime\right)
	    \end{bmatrix}=\begin{bmatrix}
	    \sigma_{\delta^\prime}\left(x,u\right)\\\sigma_{\Delta^\prime_u}\left(x,u\right)\end{bmatrix}\begin{bmatrix}\hat{W}_V\\\hat{W}_Q\\\hat{W}_R\end{bmatrix},
	\end{equation}
	where
	$
	    \sigma_{\delta^\prime}\left(x,u\right)=\Big[\left(Ax+Bu\right)^T\left(\nabla_x\sigma_V(x)\right)^T,
	  \sigma_Q(x)^T,$ $  \sigma_{R1}(u)^T\Big]$, $
	\sigma_{\Delta^\prime_u}\left(x,u\right)=\Big[B^T\!\!\left(\nabla_x\sigma_V(x)\right)^T\!\!\!,0_{m\times P},2\sigma_{R2}(u)\Big].
	$
	
	The IRL problem is then formulated as the need to estimate $\hat{W}_V,\hat{W}_Q,$ and $\hat{W}_R$ by minimizing $\delta^\prime$ and $\Delta^\prime_u.$ However, the IRL problem, as formulated above, is ill-posed, because the minimization problem $\min_{\hat{W}^{\prime}} \left|\delta'\right| + \left\Vert\Delta_u'\right\Vert $ admits an infinite number of solutions, including the trivial solution $\hat{W}_V=\hat{W}_Q=\hat{W}_R=0$ and the scaled solutions $\hat{W}_V=\alpha W_V^*, \hat{W}_Q=\alpha W_Q^*,$ and $ \hat{W}_R=\alpha W_R^*$ $ \forall \alpha\in\mathbb{R}_{>0}.$ To address the scaling ambiguity and to remove the trivial solution, a single reward weight will be assumed to be known. Since the optimal solution corresponding to a cost function is invariant with respect to arbitrary scaling of the cost function, establishing the scale by assuming that one of the weights as known is without loss of generality. Selecting $r_1$ as the known weight and removing it from \eqref{eq: Initial IRL Matrix} yields
	\begin{equation}
	    \begin{bmatrix}\delta\left(x,u,\hat{W}\right)\\\Delta_u\left(x,u,\hat{W}\right)
	    \end{bmatrix}=\begin{bmatrix}
	    \sigma_\delta\left(x,u\right)\\\sigma_{\Delta_u}\left(x,u\right)\end{bmatrix}\begin{bmatrix}\hat{W}_V\\\hat{W}_Q\\\hat{W}_R^-\end{bmatrix}+\begin{bmatrix}u_1^2r_1\\2u_1r_1\\0_{m-1\times 1}\end{bmatrix},\label{eq: Error Matrix}
	\end{equation}
	where $\hat{W}_R^-$ denotes $\hat{W}_R$ with the first element removed, $\hat{W}:=\begin{bmatrix}\hat{W}_V^T & \hat{W}_Q^T & \left(\hat{W}_R^-\right)^{T}\end{bmatrix}^T$, 
	$
	    \sigma_{\delta}\left(x,u\right)=\big[\left(Ax+Bu\right)^T\left(\nabla_x\sigma_V(x)\right)^T, \sigma_Q(x)^T, \left(\sigma_{R1}^-(u)\right)^T\big], $ and $
	\sigma_{\Delta_u}\left(x,u\right)=\begin{bmatrix}B^T\left(\nabla_x\sigma_V(x)\right)^T&0_{m\times n}&2\sigma_{R2}^-(u)\end{bmatrix},
	$
	where $\left(\sigma_{R1}^-(u)\right)^T$ and $\sigma_{R2}^-(u)$ denote $\sigma_{R1}^T(u)$ and $\sigma_{R2}(u)$ with the first columns removed.
	
    We can formulate the IRL problem as a state estimation problem by utilizing the IBE and the controller equation in an observer framework. Such a formulation allows us to address general output feedback linear systems and to leverage the use of Kalman gains under noisy conditions.
    
    To cast the IRL problem in a state estimation form, the ideal weights are concatenated with the system state to yield the concatenated state vector $z = \begin{bmatrix}x^T & \left(W^*\right)^T\end{bmatrix}^T$, where $W^*:=\begin{bmatrix}\left(W_V^*\right)^T & \left(W_Q^*\right)^T & \left(\left({W}_R^*\right)^-\right)^T\end{bmatrix}^T$. Since the ideal weights are constant, the dynamics of the concatenated state is expressed as
    $
        \dot{z} = \begin{bmatrix}Ax+Bu\\0_{2P+M-1\times 1}\end{bmatrix} ,$ and $
        y = h(z),
    $ where $y$ denotes the measurement vector and $h(z)$ is the corresponding measurement model to be designed in the following.
    
    \section{A memoryless observer (MLO)}\label{sec: PE}
    The key idea behind MLO is to treat the measurements, $y^{\prime}$, and the measured/known quantities in \eqref{eq: Error Matrix} as the \emph{output}, $y\in\mathbb{R}^{L+1+m}$, used for estimation of the concatenated state. The output is thus given by
    $ y = \begin{bmatrix}
        \left(y^{\prime}\right)^T&
      -u_1^2r_1&
      -2u_1r_1&
      0_{1\times m-1}
        \end{bmatrix}^T.$
    The corresponding measurement model is developed by using \eqref{eq: Error Matrix} to express the output as a function of the concatenated state as
    $
        h(z) = \begin{bmatrix}
        Cx\\
        \begin{bmatrix}\sigma_\delta\left(x,u\right)\\\sigma_{\Delta_u}\left(x,u\right)\end{bmatrix}\begin{bmatrix}W_V^*\\W_Q^*\\\left(W_R^*\right)^-\end{bmatrix}\end{bmatrix}.
    $
    Let $g(\hat{x},u):=\begin{bmatrix}
	    \sigma_\delta\left(\hat{x},u\right)\\\sigma_{\Delta_u}\left(\hat{x},u\right)\end{bmatrix}$ and $\sigma_u(u_1):=\begin{bmatrix}-u_1^2r_1\\
  -2u_1r_1\\
  0_{m-1\times 1}\end{bmatrix}$. The observer can then be designed as
    \begin{equation}\label{eq: Final PE Observer}
        \begin{bmatrix}\dot{\hat{x}}\\\dot{\hat{W}}\end{bmatrix}=\begin{bmatrix}A\hat{x}+Bu\\0_{2P+M-1\times 1}\end{bmatrix}+K\left(\begin{bmatrix}Cx\\\sigma_u(u_1)\end{bmatrix}-\begin{bmatrix}C\hat{x}\\g(\hat{x},u)\hat{W}\end{bmatrix}\right),
    \end{equation}
    where $K\in\mathbb{R}^{n+2P+M-1\times L+m+1}$ is the observer gain matrix, designed in the following section. 
    \subsection{Observer Gain Design and Stability Analysis}
    In the following analysis, the gain matrix $K$ will be designed in a block diagonal form. In particular, we choose $K_{MLO}:=\begin{bmatrix}K_1 & 0_{n\times 1+m}\\0_{2P+M-1\times L} & \gamma g(\hat{x},u)^TK_2\end{bmatrix}$ and $\gamma:=\nicefrac{1}{\left(\nu\Vert g(\hat{x},u)^Tg(\hat{x},u)\Vert+1\right)}$ where $\nu\in\mathbb{R}_{\geq 0}$ is a tunable constant.
    
    The following theorem analyzes the stability properties of the resulting MLO using persistence of excitation.
    \begin{definition}\label{def: PE}
	\cite{SCC.Sastry.Bodson1989} A bounded signal $t\mapsto A(t)$ is called persistently excited, if for all $t\geq 0$ there exists $\alpha_1,\alpha_2,\delta\in\mathbb{R}_{>0}$ such that\footnote{The notation $I$ denotes an identity matrix.} $
           \alpha_2 I\geq\int_{t}^{t+\delta}A^T(\tau)A(\tau)\dif\tau\geq\alpha_1 I.$
    \end{definition}

    \begin{theorem}\label{Thm: State Exponential}Provided the gain $K_1$ is selected such that $\left(A-K_1C\right)$ is Hurwitz, the gain $K_2$ is selected to be a symmetric positive definite matrix, and $ g(\hat{x},u)$ is PE, then $ \lim_{t\to\infty}\tilde{W}(t) = 0.$
    \end{theorem}
	\begin{proof}
		The dynamics for the system state estimation errors can be described by $\dot{\tilde{x}}=Ax+Bu-A\hat{x}-Bu-K_1C\tilde{x}=\dot{\tilde{x}}= (A-K_1C) \tilde{x}.$
		If $A-K_1C$ is Hurwitz, then $\tilde{x}$ converges exponentially to the origin.
		
		The dynamics of the weight estimation error can be expressed as $            \dot{\tilde{W}}=-\gamma g(\hat{x},u)^TK_2\sigma_u(u_1)+\gamma g(\hat{x},u)^TK_2g(\hat{x},u)\hat{W}. $
		Adding $\pm \gamma g(\hat{x},u)^TK_2g(\hat{x},u){W}^*$ and using the fact that $\sigma_u(u_1)=g(x,u)W^*$, the weight estimation error dynamics can be expressed as a perturbed linear time-varying system
		\begin{equation}\label{eq: Perturbed sys}
			\dot{\tilde{W}}=-A(t)\tilde{W}+B(t),
		\end{equation}
		where $A(t):=\gamma(t) g(\hat{x}(t),u(t))^TK_2g(\hat{x}(t),u(t))$ and $B(t):=\gamma(t) g(\hat{x}(t),u(t))^TK_2(g(\hat{x}(t),u(t)) - g(x(t),u(t)))W^*$.
		Since $\hat{x},x,u\in\mathcal{L}_\infty,$ Theorem 2.5.1 from \cite{SCC.Sastry.Bodson1989} implies that the nominal system $\dot{\tilde{W}}=-A(t)\tilde{W}$ is globally exponentially stable (GES) if $K_2$ is a symmetric positive definite matrix and the signal $(\hat{x},u)$ is PE.
		
	   Lemma 4.6 from \cite{SCC.Khalil2002}, can then be invoked with $B$ as the input and $\tilde{W}$ as the state to conclude that $\dot{\tilde{W}}=-A(t)\tilde{W}+B(t).$
	\end{proof}

\section{Inclusion of memory (HSO)}\label{sec: Hist}
The observer designed in the previous section relies on \textit{persistent} excitation for stability and convergence. As a result, it suffers from the well-known lack of robustness of PE-based adaptive control methods under loss of excitation. This section develops an observer (called the HSO) that relies on re-use of previously recorded data (henceforth referred to as the history stack) for robustness. If the system trajectories are PE, then the HSO results in convergence of the estimation errors to the origin, similar to the MLO. However, as opposed to the MLO, through the use of a history stack, the HSO guarantees boundedness of the state estimation errors even under loss of excitation.

The output for the HSO is
$
    y(t) = \Big[\!
    \left(y^{\prime}(t)\right)\!^T\!,
   -u_1^2(t_1)r_1,
  -2u_1(t_1)r_1,
  0_{1\times m-1},
  \hdots, 
  -u_1^2(t_N)r_1, $
  $-2u_1(t_N)r_1, 
  0_{1\times m-1}
    \Big]^T,
$ with the corresponding measurement model, obtained by using past control values and past state estimates in \eqref{eq: Error Matrix}, given by
\begin{equation*}
    \hspace*{-1cm}h({z}) = \begin{bmatrix}
    C{x}\\
    \begin{bmatrix}
	    \sigma_\delta\left({x}(t_1),u(t_1)\right)\\\sigma_{\Delta_u}\left({x}(t_1),u(t_1)\right)\\\vdots\\\sigma_\delta\left({x}(t_N),u(t_N)\right)\\\sigma_{\Delta_u}\left({x}(t_N),u(t_N)\right)\end{bmatrix}\begin{bmatrix}W_V^*\\W_Q^*\\\left(W_R^*\right)^-\end{bmatrix}\end{bmatrix},
    \end{equation*}
    where $\sigma_\delta\left({x}(t_i),u(t_i)\right)$ and $\sigma_{\Delta_u}\left({x}(t_i),u(t_i)\right)$ denotes $\sigma_\delta\left({x}(t),u(t)\right)$ and $\sigma_{\Delta_u}\left({x}(t),u(t)\right)$ evaluated at time $t_i$, respectively.
    
    It is assumed that at every time instance $t$, the observer has access to a history stack $\mathcal{H} \coloneqq \left\{\hat{\Sigma},\Sigma_u\right\}$, defined as
    \[
        \hat{\Sigma}:=\begin{bmatrix}
        \sigma_\delta\left(\hat{x}(t_1),u(t_1)\right)\\\sigma_{\Delta_u}\left(\hat{x}(t_1),u(t_1)\right)\\\vdots\\\sigma_\delta\left(\hat{x}(t_N),u(t_N)\right)\\\sigma_{\Delta_u}\left(\hat{x}(t_N),u(t_N)\right)\end{bmatrix}, \ \ \ \Sigma_u:=\begin{bmatrix}-u_1^2(t_1)r_1\\
        -2u_1(t_1)r_1\\
        0_{m-1\times 1}\\
        \vdots\\
        -u_1^2(t_N)r_1\\
        -2u_1(t_N)r_1\\
        0_{m-1\times 1}\end{bmatrix},
    \]
    where time instances $t_1,\ldots,t_N$ are selected to ensure that the resulting history stack is full rank, as subsequently defined in Def. 2.
    Denoting the observer gain matrix by $K\in\mathbb{R}^{n+2P+M-1\times L+N(1+m) }$, the HSO is designed as
    \begin{equation}\label{eq: Final observer}
        \begin{bmatrix}\dot{\hat{x}}\\\dot{\hat{W}}\end{bmatrix}=\begin{bmatrix}A\hat{x}+Bu\\0_{2P+M-1}\end{bmatrix}+K\left(\begin{bmatrix}Cx\\\Sigma_u\end{bmatrix}-\begin{bmatrix}C\hat{x}\\\hat{\Sigma}\hat{W}\end{bmatrix}\right).
    \end{equation}

    The error in equation \eqref{eq: Initial IRL Matrix} implies that the innovation $\Sigma_u - \hat{\Sigma}\hat{W}$ in \eqref{eq: Final observer} corresponds to the weight estimation error $\tilde{W}$ only if $\hat{\Sigma} = \Sigma$. Since $\hat{\Sigma}$ depends continuously on $\hat{x}$ and because $\hat{x}$ exponentially converges to $x$, $\hat{\Sigma}$ exponentially converges to $\Sigma$. As a result, newer and better estimates of $x$ can be leveraged to improve the estimates of $W^*$ by purging and refreshing the history stack $\mathcal{H}$. Due to purging, the time instances $\left\{t_1,\cdots t_N\right\}$ and the matrices $\hat{\Sigma}$ and $\Sigma_u$ are piecewise constant functions of time. 
    
    \begin{definition}\label{def: FUll Rank}
    The history stack is called \emph{full rank} if $\rank\left(\hat{\Sigma}\right) = 2P+M-1$. The signal $\left(\hat{x},u\right)$ is called \emph{finitely informative} (FI) if there exist time instances $0\leq t_1 < t_2 < \cdots < t_N$ such that the resulting history stack is full rank and \emph{persistently informative} (PI) if for any $T\geq 0$, there exist time instances $T\leq t_1 < t_2 < \cdots < t_N$ such that the resulting history stack is full rank.
    \end{definition}
    A history stack management algorithm similar to \cite[Fig. 1]{SCC.Kamalapurkar2018} is used to ensure the existence of a time instance $t_M$ such that, if the signal $(\hat{x},u)$ is FI, then the history stack is full rank for all $t\geq t_M$, and in addition, if it is PI, then $\lim_{t\to \infty}\left\Vert \Sigma(t) - \hat{\Sigma}(t)\right\Vert = 0$.
    
    \subsection{Observer Gain Design and Stability Analysis}
 
    The HSO gain matrix is designed in the block diagonal form
    $ K_{HSO}\coloneqq\begin{bmatrix}K_3 & 0_{n\times N+Nm}\\0_{2P+M-1\times L} & K_4\left(\hat{\Sigma}^T\hat{\Sigma}\right)^{-1}\hat{\Sigma}^T\end{bmatrix},\label{eq: final HS gain}$
    where $K_3\in\mathbb{R}^{n\times L}$ is a constant gain matrix and $K_4:\mathbb{R}_{\geq 0}\to \mathbb{R}^{2P+M-1\times 2P+M-1}$ is a potentially time-varying gain matrix. Provided the gain matrices are selected to satisfy the hypothesis of Theorem~\ref{thm:2} below, the resulting observer in \eqref{eq: Final observer} can be shown to be convergent in the presence of PE and bounded under loss of excitation. Finite excitation is needed for the history stack to be full rank so that $(\hat{\Sigma}^T\hat{\Sigma})^{-1}$ is well-defined.
    
    \begin{theorem}\label{thm:2}
    Provided $K_3$ is selected such that $\left(A-K_3C\right)$ is Hurwitz, $K_4(t)$ is selected such that for $t<t_M$, $K_4(t) = 0$ and for $t\geq t_M$, $K_4(t)$ is symmetric positive definite, $0 < \underline{k} \leq \inf_{t\geq t_M} \left\{\lambda_{\min}K_4(t)\right\} $ and $\sup_{t\geq t_M} \left\{ \left\Vert K_4(t)\right\Vert\right\}  \leq \overline{k}<\infty$, then $ \tilde{W} $ is ultimately bounded (UB) if the signal $\left(\hat{x},u\right)$ is FI and $\lim_{t\to\infty}\tilde{W}(t) = 0$ if it is PI.
    \end{theorem}
    \begin{proof}
    Using Theorem \ref{Thm: State Exponential}, if $\left(A-K_3C\right)$ is Hurwitz, $\tilde{x}(t)\to 0$ exponentially as $t\to\infty.$
    Using \eqref{eq: Final observer}, the dynamics of the weight estimation error can be expressed as $ \dot{\tilde{W}} = K_4(t)\hat{W}-K_4(t)\left(\hat{\Sigma}^T(t)\hat{\Sigma}(t)\right)^{-1}\hat{\Sigma}^T(t)\Sigma_u (t)$. Since $K_4$ is set to $0$, the weight estimates are constant over $[0,t_M)$. For $t\geq t_M$, adding $\pm K_4(t)W^*$ to $\dot{\tilde{W}}$, and using the fact that $\Sigma W^*=\Sigma_u$, the weight estimation error dynamics can be treated as the controlled system
    \begin{equation}
        \dot{\tilde{W}} = -K_4(t)\tilde{W}+K_4(t)w,\label{eq: Closed Loop Error Dynamics}
    \end{equation}
     where $w(t) \coloneqq \left(I-\left(\hat{\Sigma}^T(t)\hat{\Sigma}(t)\right)^{-1}\hat{\Sigma}^T(t)\Sigma(t)\right)W^*$ is treated as the control input. Using the Cauchy-Schwartz Inequality and the Rayleigh-Ritz Theorem \cite{SCC.Horn.Johnson1993}, the orbital derivative of the positive definite candidate Lyapunov function $ V\left(\tilde{W}\right)\coloneqq\frac{1}{2}\tilde{W}^T\tilde{W}$ along the trajectories of \eqref{eq: Closed Loop Error Dynamics} can be bounded as
    \begin{equation}
        \dot{V}\left(t,\tilde{W}\right)\leq-\underline{k}\left\Vert\tilde{W}\right\Vert^2+\overline{k}\left\Vert\tilde{W}\right\Vert\left\Vert w\right\Vert,
	\end{equation}
	for all $t\geq t_M$ and $\tilde{W}\in\mathbb{R}^{2P+M-1}$.
    In the domain $\Vert \tilde{W}\Vert>\frac{2\overline{k} \Vert W^*\Vert}{\underline{k}}\left\Vert w \right\Vert$, the orbital derivative satisfies the bound $ \dot{V}\left(t,\tilde{W}\right) \leq -\frac{\underline{k}}{2}\Vert\tilde{W}\Vert^2$.
    Using Theorem 4.19 from \cite{SCC.Khalil2002}, it can be concluded that the controlled system in \eqref{eq: Closed Loop Error Dynamics} is ISS. 
    
    If the signal $(\hat{x},u)$ is PI, then the history stack can be purged and refreshed infinitely many times such that $w(t)\to 0$ as $t\to\infty$. Utilizing Exercise 4.58 from \cite{SCC.Khalil2002}, it can then be concluded that $\tilde{W}(t) \to 0$ as $t\to\infty$.
    
    If the signal $(\hat{x},u)$ is FI but not PI, then there exists a time instance $T$ such that the history stack remains unchanged for all $t\geq T$. As a result, there exists a constant $\overline{w}$ such that for all $t\geq T$, $\left\Vert w(t) \right\Vert \leq \overline{w}$. By the definition of ISS, it can then be concluded that $\tilde{W}$ is UB.
    \end{proof}
    
    \begin{remark}
    The UB result in the absence of PE is a distinct advantage of HSO over MLO, which provides no such guarantee. Once the system states are no longer exciting, the MLO could potentially become unstable.
    \end{remark}
    \begin{remark} The IRL-O formulation is not restricted to the choices of $K$ in Theorems \ref{Thm: State Exponential} and \ref{thm:2}. Different stabilizing or heuristic gain selection methods can  be incorporated in the developed framework. For example, motivated by robustness to measurement noise, the use of a Kalman filter for gain selection is explored in Section~\ref{sec: Sim}.
    \end{remark}
    
    \section{Simulations}\label{sec: Sim}
    A key motivation for casting the IRL problem into the observer framework is that the observer can be extended to a Kalman filter in a straightforward fashion to address measurement noise. To implement the developed observers as Kalman filters, all that is needed is to select the gains $K_3$ and $K_4$ using the Kalman gain update equations. The following simulation study demonstrates the validity, the robustness, and the performance of the designed observers and their Kalman filter implementation.
    
    While the developed observer IRL methods are applicable to general output feedback linear systems, the concurrent learning (CL) method used for comparison is only applicable to a restricted set of systems (the state estimator in \cite{SCC.Kamalapurkar2017} is modified slightly for the non-Brunovsky form of \eqref{eq: Sim Dyns}). In the following, to make comparisons feasible, a system that both methods are applicable to is selected.
    
    The agent under observation has linear dynamics \begin{equation}
        \dot{x}=\begin{bmatrix}2 & 1\\3 & 2\end{bmatrix}x+\begin{bmatrix}2\\0.5\end{bmatrix}u, \  y = \begin{bmatrix}1 & 0\end{bmatrix}x.\label{eq: Sim Dyns}
    \end{equation}
    The optimal controller, 
    $
        u^* (x)= -\begin{bmatrix}4.14 & 5.53\end{bmatrix}x,
    $
    minimizes an LQR problem, $Q=\text{diag}([2,11])$ and $R=1.5,$
    with an optimal value function
    $
        V^*(x)=2.54x_1^2+7.59x_2^2+4.50x_1x_2.
    $
    The ideal weights that are to be estimated are ${W}^*_{V1}=2.54,{W}^*_{V2}=7.59,{W}^*_{V3}=4.50,{W}^*_{Q1}=2,{W}^*_{Q2}=11,$ and $R=1.5$ is selected as the known value to remove the scaling ambiguity.
    
    Since the system state estimates converge exponentially to the true system states, a time based purging technique similar to \cite[Fig. 1]{SCC.Kamalapurkar2018} is utilized to reduce the estimation error associated with the system state estimates stored in the history stack. Furthermore, to improve numerical stability of gain computation, the history stack management algorithm also attempts to minimize the condition number of $\hat{\Sigma}^T\hat{\Sigma}$. In the presented simulation studies, the history stacks contain data for five previous time instances and are purged every $0.5$ seconds if they can be repopulated.
    
    Three simulation studies are performed. The first shows the performance of the designed observers in a noise-free setting for a system with two states and one control. The second simulation shows the designed observers in a noise-free setting for a larger dimensional system. The last simulation incorporates noise in order to investigate the robustness of the observer.
    
    The error metric used to compare all of the observers/filters is the summation of the five relative weight estimation errors, defined as

    {
		\medmuskip=0mu
		\thinmuskip=0mu
		\thickmuskip=0mu
		\begin{equation*}\label{eq: Metric}
		\sum\frac{\tilde{W}_i}{W^*_i}:=\frac{\Vert\tilde{W}_{V_1}\Vert}{W_{V_1}^*}+
		\frac{\Vert\tilde{W}_{V_2}\Vert}{W_{V_2}^*}+
		\frac{\Vert\tilde{W}_{V_3}\Vert}{W_{V_3}^*}+\frac{\Vert\tilde{W}_{Q_1}\Vert}{W_{Q_1}^*}+
		\frac{\Vert\tilde{W}_{Q_2}\Vert}{W_{Q_2}^*}.
		\end{equation*}
	}
 
    \subsection{PE Signal without Noise }\label{sec: PE no Noise}
    \subsubsection{Two State System}\label{sec: Two States}
    The first simulation study concerns a noise-free environment. The controller that the agent under observation implements is a combination of the optimal controller, $u^*$, and a known additive excitation signal, i.e., the feedback controller of the agent is
    $ u(t,x(t)) =u^*(x(t))+u_{exc}(t), $
    where $u_{exc}(t):=5\sin(t)+18\cos(0.4t)+36\sin(2t)+0.5\cos(3t)$ induces excitation in the signal $\hat{x}.$
    
    The HSO in \eqref{eq: Final observer}, is implemented using three different $K_{HSO}$ matrices, comprised of the same $K_3$ matrices, computed using the ``place" command in MATLAB for poles $p_1=-2$ and $p_2=-4$, and three different $K_4$ matrices. The first two $K_4$ matrices are computed using gains $K_4=-I$ and $K_4-0.5I$ (denoted in Fig. \ref{fig: NoNoise_PE} as HSO - P $= -1$ and HSO - P $= -0.5$, respectively). The third $K_4$ matrix is selected to be an exponentially varying gain matrix, $K_4=(1-0.9\exp^{-t})0.5I$ (denoted as HSO - Exp in Fig. \ref{fig: NoNoise_PE}). The MLO in \eqref{eq: Final PE Observer} is implemented using a single $K_{MLO}$ matrix, with $K_1$ computed using the ``place" command for poles $p_1=-2$ and $p_2=-4$, and $K_2=10000I.$
    
    As seen in Fig. \ref{fig: NoNoise_PE}, all of the weight estimation errors for the designed observers converge to the origin as expected. Even though there is a larger initial estimation error for the HSO, with constant gains, compared to the MLO, the history stack based observers converge much quicker than  the MLO. The initial estimation error can be reduced for the HSO either by moving the poles closer to the origin, or implementing an exponentially varying gain matrix, as in the HSO-Exp case. The exponentially varying gain matrix combines the benefits of initial small gains, when the state estimates are inaccurate, with those of progressively larger gains, leading to fast convergence.
    
    \begin{figure}
        \centering
        \includegraphics[clip,trim = 0.5cm 10cm 0.5cm 10cm, width=0.5\textwidth]{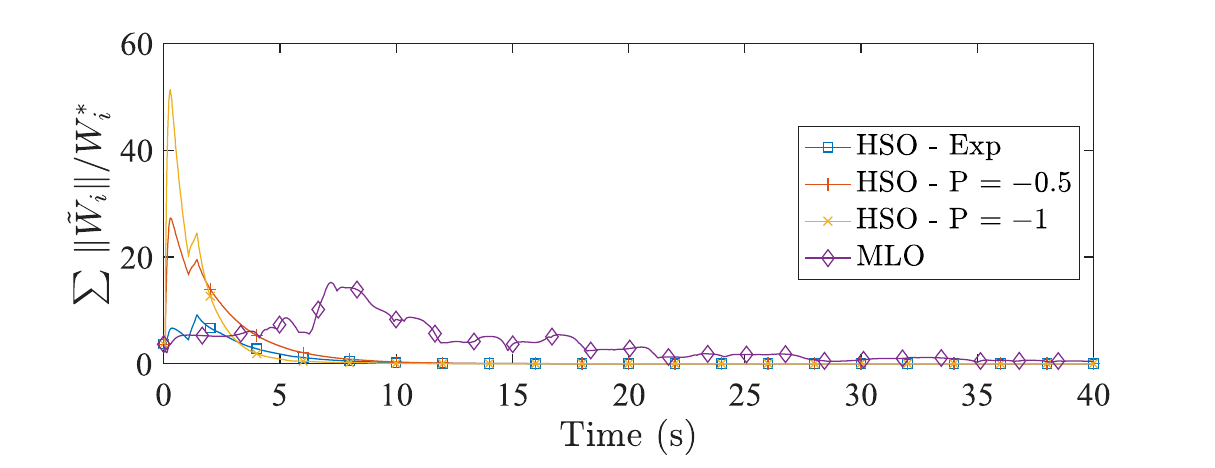}
        \caption{Weight estimation errors for the developed observers with no noise and PE signal.}
        \label{fig: NoNoise_PE}
    \end{figure}
    
    \subsubsection{Four State System}
    The second simulation shows a four state system with the exponentially varying HSO and the Kalman filter implementation of the HSO observer in a noise-free setting. Similarly to Section \ref{sec: Two States}, the agent under observation implements a combination of the optimal controller and a known exciting controller. In this simulation, the exciting controller is randomly selected from a uniform distribution in the set $[0,10]$. The dynamical system of the agent under observation is
    \begin{equation}
        \dot{x}=\begin{bmatrix}
            2&4&1&0\\0&3&6&2\\3&2&2&6\\3&5&6&2
        \end{bmatrix}x+\begin{bmatrix}
            7&2\\4&5\\3&3\\2&6
        \end{bmatrix}u, \ y = \begin{bmatrix}1 &0& 0& 0\\0&1& 0& 0\\0& 0& 0& 0\\0& 0& 0& 0\end{bmatrix}x.
    \end{equation}
    The optimal controller, \[u^*(x)=\begin{bmatrix}1.79&	-0.235&	1.022&	0.487\\
-0.345&	1.75&	3.96&	4.35\end{bmatrix}x,\]minimizes an LQR problem, $Q = \text{diag}([2,5,8,11]),R = \text{diag}([1.5,0.5])$, and $R(1,1) = 1.5$ is selected as the known weight.

As seen in Figure \ref{fig: NoNoise_PE4SS}, both the exponential and Kalman filter implementations of the HSO converge to the origin. In Figure \ref{fig: NoNoise_PE4SSPE}, the MLO eventually estimates the unknown weights in the reward and value function, however, the MLO takes significantly longer than the HSO implementations. This further validates the design for using previously recorded data to update the weight estimates.
    \begin{figure}
        \centering
        \includegraphics[width=0.5\textwidth]{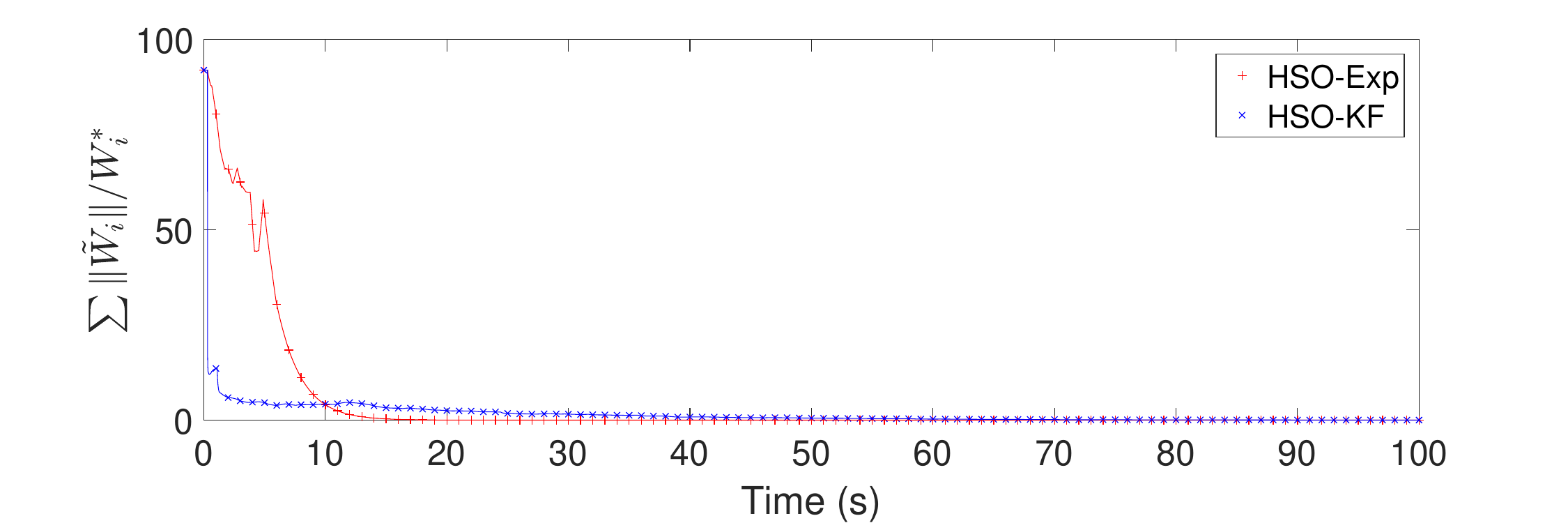}
        \caption{Weight estimation errors for the developed HSO observers with no noise and PE signal with larger dimensional system.}
        \label{fig: NoNoise_PE4SS}
    \end{figure}
    
    \begin{figure}
        \centering
        \includegraphics[width=0.5\textwidth]{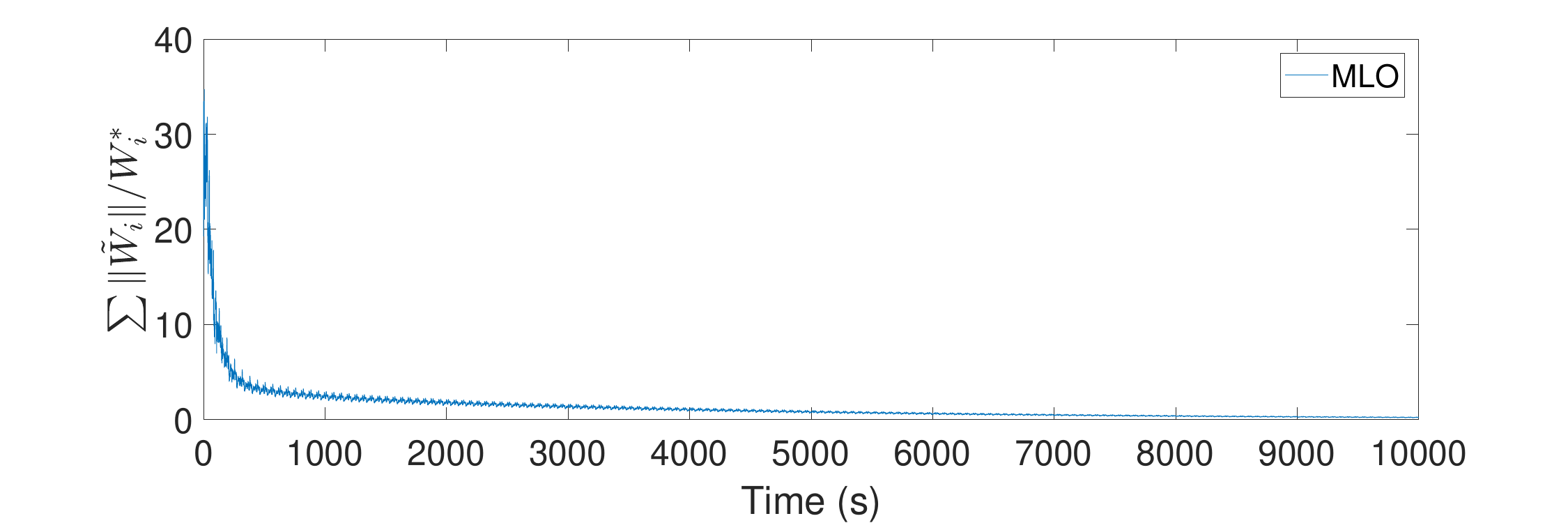}
        \caption{Weight estimation errors for the developed MLO observer with no noise and PE signal with larger dimensional system.}
        \label{fig: NoNoise_PE4SSPE}
    \end{figure}
    
    \subsection{Persistently Excited Signal with Noise}
    The last simulation is an investigation into noise robustness of the HSO and the Kalman filter implementation of the HSO (called HSO-KF) compared to the CL update law in \cite{SCC.Self.Abudia.ea2020,SCC.Self.Mahmud.ea2020}. The simulation for comparison is the same two state linear system as in Section \ref{sec: Two States}. The state estimator used for the CL method is developed in \cite{SCC.Kamalapurkar2017} (with a slight modification to address the non-Brunovsky form of the dynamics). Zero-mean Gaussian noise is added to $y'$ and $u$, with three noise variances used to simulate low-noise ($R_1 = \text{diag}([0.1^2, 0.1^2])$), medium noise ($R_2 = \text{diag}(0.5^2, 0.5^2])$ and high noise ($R_3 = \text{diag}([1^2, 1^2])$) scenarios. Fifty Monte-Carlo simulations for each noise level are conducted and compared with the no-noise case. We do not study the behavior of the MLO under noisy measurements due to the added robustness of the HSO due to the use of past data.
    
    The results of the simulation study are shown in Table \ref{Tab: Noise}. As seen from the data, all three methods perform well in the noise free case, and the performance of all three methods is comparable in the low noise scenario. The advantages of the two HSO methods over the CL method are evident in the medium and high noise scenarios. Both the HSO-Exp and HSO-KF show better robustness to noise when compared to the CL method, especially in high the noise situation (CL steady state (SS) error is almost four times higher than both HSO methods). Comparing the results of HSO-Exp to HSO-KF, HSO-Exp has lower SS errors for the low and medium noise cases, while, HSO-KF has lower SS errors for the no noise and high noise cases. In addition, HSO-KF converges quicker in every case compared to both CL and HSO-Exp, as evidenced by the average over the whole time interval (TT).
    
\begin{table}[h]
    \caption{Comparison between concurrent learning (CL), KF based implementation of HSO (HSO-KF), and exponential pole selection implementation of HSO (HSO-Exp), with different noise variances. Simulations were ran for 100 seconds over 50 trials with step size $Ts=0.005.$ The standard deviations (SD) simulated are $0.0,0.1,0.5,$ and $1.0$. The metric used for comparison is the average of the average on the trajectories $\sum \tilde{W}_i/W_i^*$, where TT denotes the average over the entire trajectory, and SS denotes the average over the last $30$ seconds of the trajectory. The exponential HSO gains are selected similar to Section \ref{sec: PE no Noise}, except $K_4=(1-0.9\exp^{-t})0.15I$. The Kalman filter gain is selected using the gain matrix $K_{HSO}=\text{diag}([K_3,K_4])$ where $K_3$ and $K_4$ are independent Kalman gains.}
    \label{Tab: Noise}
    \begin{center}
	    \begin{tabular}{|c|cc|cc|cc|}
	    \hline
	         & \multicolumn{2}{c|}{CL} & \multicolumn{2}{c|}{HSO-Exp}
	        & \multicolumn{2}{c|}{HSO-KF}\\\hline
	          SD & TT & SS&TT & SS & TT & SS\\
	          \hline
	           0.0 & 0.9855 &7.43e-05 &0.9101 &8.41e-05 &0.0446 &1.86e-14 \\
	           \hline
	            0.1 & 0.8977& 0.2647&0.8463 &0.1652 &0.2591 &0.2279 \\
	            \hline
	            0.5 & 2.1766& 2.0336& 1.3064& 0.6894& 0.7291&0.7277 \\
	             \hline
	            1.0 & 5.5415 & 5.5223&1.9055 &1.4667 &1.5055 &1.4111\\
	             \hline
	    \end{tabular}
	\end{center}
\end{table}
    
    \section{Conclusion}\label{sec: Con}
    This paper presents a novel observer-like formulation for performing online estimation of reward functions using input-output observations. Two observers are proposed and their convergence guarantees are established. The Monte-Carlo simulations demonstrate that the developed observer based IRL techniques, utilizing exponentially varying gains and Kalman gains, demonstrate better noise robustness than existing CL based IRL techniques \cite{SCC.Self.Abudia.ea2020,SCC.Self.Mahmud.ea2020}.
    
	Future work includes extension of the observer-based IRL methods to nonlinear systems and investigation of data storage algorithms for maintaining informative data in the history stack.

	\bibliographystyle{IEEEtran}
	\bibliography{mybib}

    \end{document}